\documentclass[]{svjour3}
\usepackage{graphicx}
\usepackage{epstopdf}
\usepackage{amsmath,amssymb,amsfonts}
\usepackage{mathtools}
\usepackage{multirow}
\usepackage{verbatim}
\usepackage{url}
\usepackage{comment}
\usepackage{mathtools}
\usepackage{epsf,pgf,graphicx}
\usepackage{color}
\usepackage{tikz,pgf}
\usepackage{diagbox}
\usepackage{makecell}
\usepackage{color}
\usepackage{braket}
\usepackage{mathtools}
\usepackage{xcolor}
\usepackage{subfig}
\usepackage{hyperref}
\usepackage{rotating}

\usepackage{tabstackengine}
\stackMath

\usepackage{array}
\newcolumntype{M}[1]{>{\hbox to #1\bgroup\hss$}l<{$\egroup}}

\makeatletter
\newcommand\@brcolwidth{0.67em}

\def\@brarray[#1]{\array{r*\c@MaxMatrixCols {M{#1}}}}
\makeatother

\title{Analysis of Boolean Functions Related to Binary Input Binary Output Two-party Nonlocal Games}
\author{Jyotirmoy Basak \and Subhamoy Maitra \and Prabal Paul \and Animesh Roy}
\institute{
	J. Basak \at
	Applied Statistics Unit, Indian Statistical Institute, Kolkata, India.\\
	\email{bjyotirmoy.93@gmail.com}
	\and
	S. Maitra \at
	Applied Statistics Unit, Indian Statistical Institute, Kolkata, India.\\
	\email{subho@isical.ac.in}
	\and
	P. Paul \at
	Department of Mathematics, BITS-Pilani K K Birla Goa Campus, India.\\
	\email{prabal.paul@gmail.com}
	\and
	A. Roy \at
	Applied Statistics Unit, Indian Statistical Institute, Kolkata, India.\\
	\email{animesh.roy03@gmail.com}
}
\begin{document}
\maketitle
\begin{abstract}
	The famous CHSH game can be interpreted with Boolean functions while understanding the success probability in the classical scenario. In this paper, we have exhaustively studied all the Boolean functions on four variables to express binary input binary output two-party nonlocal games and explore their performance in both classical and quantum scenarios. Our analysis finds out some other games (other than the CHSH game) which offer a higher success probability in the quantum scenario as compared to the classical one. Naturally, our study also notes that the CHSH game (and the games corresponding to the similar partition) is the most efficient in terms of separation between quantum and classical techniques.
\end{abstract}




\section{Introduction}\label{sec1}

Quantum computation is strikingly more powerful than classical computation and this is evident from the evaluation of quantum algorithms which can be exponentially faster~\cite{SH94,SM94} as compared to the conventional classical algorithms. Because of this potency of quantum computation, quantum cryptography~\cite{BB84} offers additional security that is impossible to replicate in the classical world. This kind of quantum advantage can also be achieved in the case of nonlocal games.

Nonlocal games refer to the games played between multiple space-separated players and a referee where communication between the players is strictly forbidden during the game. In a binary input binary output two-party nonlocal game, the referee sends an input bit to each of the players who then respond by sending output bits to the referee. Based on the winning condition, the players fix some strategies among themselves in the classical scenario before the game begins. Similarly in the quantum scenario, the players share some entanglement among themselves before the start of the game to get some advantage in the winning probability as compared to the classical scenario.

In a binary input binary output two-party nonlocal game, each player has two choices for the input and two choices for the output. The most well known binary input binary output two-party nonlocal game is the CHSH game~\cite{CHSH69}, where a referee provides two uniformly random bits $x_1$ and $x_2$ to each of the two players. After receiving the inputs, the two parties send their output bits $x_3$ and $x_4$ to the referee. The function that represents the CHSH game is of the form $f(x_1, x_2, x_3, x_4) = (x_1 \wedge x_2) \oplus (x_3 \oplus x_4)$. From the winning condition of the CHSH game, one can easily check that the two parties can win the game whenever the values of $x_1, x_2, x_3, x_4$ satisfy $f(x_1, x_2, x_3, x_4) = 0$. It is well known that the maximum success probability of the CHSH game in the classical scenario is $0.75$ whereas the maximum success probability using quantum resources is $\cos^2\frac{\pi}{8}$ (which is approximately $0.85$).

There are several known two-party nonlocal games that offer quantum advantages~\cite{BBT05,CM14}. However, the inputs and the outputs for any of those games are not restricted to bits. To the best of our knowledge, the CHSH game is the only known binary input binary output two-party nonlocal game that offers a quantum advantage.

The nonlocal games are interesting because for some of those games, the quantum advantage or a separation (an advantage to the maximum quantum success probability as compared to the maximum classical one) can be achieved which is often useful to prove the quantumness of a system and to certify the untrusted devices involved in a scheme in Device Independent (DI) scenario. In general, the DI certification of quanutm cryptographic schemes have been done~\cite{BCMM22,MPR17,VV14} considering the CHSH game. Recently, DI certification has also been done in Quantum Key Distribution scenario considering the three party pseudo-telepathy game~\cite{BMM19}.


Although there are several known nonlocal games that offer quantum advantage~\cite{BBT05,BBT03}, to the best of our knowledge, from the class of all possible binary input binary output two-party nonlocal games, the only known game that offers quantum advantage is the CHSH game. As the CHSH game can't be won with certainty in the quantum scenario, it would be interesting to check whether there exists any other binary input binary output two-party nonlocal game for which quantum advantage can be achieved and the game can be won with a better quantum success probability than the CHSH game (because from the analysis of~\cite{JS18}, it is clear that if there exists any such game then it can be used for DI testing instead of the CHSH game to reduce the overall sample size).  

In this article, we have explored the performance of
all possible binary input binary output two-party nonlocal games
(having atleast one successful outcome for each possible input)
by considering them as four variable boolean functions.

\subsection{Contribution and Organization}

In the current article, we analyze the performance of all those binary input binary output two-party nonlocal games (in both classical and quantum scenarios) which have atleast one successful outcome for every possible input. In Section~\ref{preli}, we begin with some preliminary discussions about our notational assumptions and introduces some definitions related to our analysis. In Subsection~\ref{incos}, we briefly describe different groups of strategies for our classical analysis and mention the structure of games related to the $2+2$ partition which can't be won with certainty. In the next Subsection (i.e., in Subsection~\ref{bascresult}), we derive some basic results that will be required for the performance analysis of different partitions of games. Next we mention the strategies to find out the maximum classical and maximum quantum success probabilities for the games corresponding to each partition. In Section~\ref{error}, we discuss the minor errors in the count values of boolean functions mentioned in table 1 of \cite{NGA19}. Finally in
Section~\ref{main}, we present the detailed analysis (both in classical and quantum scenarios) for the games corresponding to each of the partitions. Our main contributions in this paper are twofold which is enumerated below.

\begin{enumerate}
	
	\item The CHSH game is the most well-known game from the class of all possible binary input binary output two-party nonlocal games that offer quantum advantage. In this direction, here we have considered all possible binary input binary output two-party nonlocal games which have at least one successful outcome for every possible input, group them in terms of partitions of the number of successful outcomes and analyze their performance to identify whether there exist any such different game which offer quantum advantage.
	
	\item The CHSH game is also used for DI certification. In~\cite{JS18}, it is mentioned that the number of samples
	required for DI testing is inversely proportional to the success probability of the underlying nonlocal
	game and the maximum success probability of the CHSH game in the quantum scenario is around
	$0.85$. To reduce the overall sample size, we have
	explored the performance of all other games to check whether
	there exists any other game for which quantum advantage can
	be achieved.
	
\end{enumerate}

We conclude the paper in Section~\ref{concl} with directions for future research. Before proceeding further, let us first define our notational assumptions and a few definitions that are required for our analysis.

\section{Preliminaries}
\label{preli}
Every two-party nonlocal game with the input bits (say) $x_1$ and $x_2$ and the output bits (say) $x_3$ and $x_4$ can be represented as a $4$-variable boolean function (with variables $x_1, x_2, x_3$ and $x_4$). 

In the classical scenario of a nonlocal game, the players fix some strategies among themselves before the game begins. A strategy for a player may be either input-dependent or input-independent. One can easily check that for every input, a player can have exactly two input-dependent strategies (i.e., either the input bit itself or the complement of the input bit) and exactly two input-independent strategies (i.e., either output $0$ or output $1$ irrespective of the input bits). So, it is obvious that the two players can have a maximum of $16$ different strategies in the classical scenario. It is also evident that for a particular value of the input pair $x_1x_2$, there may have atmost four possible values of the output pair $x_3x_4$. The players may not win the game for all four possible values of the output pair $x_3x_4$. Without loss of generality, here we assume that for a particular assignment to the values of $x_1,x_2,x_3,x_4$, if one can win the game then the corresponding output of the boolean function is $0$, otherwise the output is $1$.

Based on the distribution of the successful outcomes (i.e., the distribution of $0'$s) in the output column of the boolean function, a binary input binary output two-party nonlocal game can be represented in terms of \emph{partitions} of the total number of successful outcomes.

\definition{
	\label{defpart}
	\indent
	{\it (Partition of a nonlocal game):~A partition is a representation of a class of $n$ party nonlocal games depending on the total number of successful outcomes. A partition of a nonlocal game is generated by splitting up the total number of successful outcomes into $2^n$ parts depending on the number of successful outcomes for each of the $2^n$ possible inputs. For an $n$-party binary input binary output nonlocal game with $d$ number of successful outcomes (where $2^n \leq d \leq 2^{2n}$), the corresponding partition will be represented as a summation of $2^n$ non-zero numbers (like $n_1 + n_2 + \cdots+ n_{2^n}$) such that $d=\sum_{i=1}^{2^n} n_i$ where each $n_i$ is the number of successful outcomes for the $i$-th input such that $0< n_i \leq 2^n$.} }\\

For a binary input binary output two-party nonlocal game, there are four possible inputs and for every input, there can have atmost four possible successful outcomes. So for these games, the partition representation is of the form $p_1+p_2+p_3+p_4$ where each $p_i$ denotes the total number of successful outcomes for the $i$-th input such that $0 \leq p_i \leq 4$. For example, one may consider the CHSH game (which represents a balanced $4$-variable boolean function) for which the partition representation is of the form $2+2+2+2$. Similarly every other binary input binary output two-party nonlocal games can be represented as a summation of four non-zero numbers. 

From these discussions, one can easily understand that many different games have the same representation of the partition. However, all the games that belong to a particular partition may not behave similarly. Here in this present effort, we are interested in discovering all those games which offer quantum advantage (i.e., a better winning probability in the quantum scenario as compared to the classical one).

\definition{
	\label{defsep}
	{\it (Separation for a nonlocal game):~A separation denotes the difference between the maximum classical and the maximum quantum success probabilities for those games which offer a quantum advantage.}}\\

For the sake of simplicity, from now onwards we use the notation $x$ and $y$ to denote input bits and the notation $a$ and $b$ to denote the output bits of the two parties. $\overline{x}, \overline{y}, \overline{a}, \overline{b}$ denotes the usual complements (bit complement) of $x, y, a, b$ respectively. Later on, if nothing is specified explicitly, whenever we use $xy$ as input and $ab$ as output for the two players, we assume that $xy$ and $ab$ can take any values from the set $\{00, 01, 10, 11\}$.

\subsection{Inconsistency for the $2+2$ Partition and Its Subpartitions} 
\label{incos}

It is well known that for a binary input binary output two-party nonlocal game, there are $4$ possible inputs and for each input, there can have atmost $4$ different outputs. It is also clear that for a particular input string (i.e., for a particular value of $xy$), the two players can have atmost $16$ different strategies to generate their outcomes in the classical scenario. Based on the outcomes, here we classify the $16$ different strategies into four groups where each group has $4$ different strategies and each of these strategies leads to a different outcome for a particular input. These four groups are as follows.\\

\textbf{Group 1 (Constant Strategies):} $00, 01, 10, 11$

\textbf{Group 2 (Input-dependent Strategies):} $xy, \overline{x}y, x\overline{y}, \overline{x}\overline{y}$

\textbf{Group 3 (Mixed Strategies):} $x0, x1, \overline{x}0, \overline{x}1$

\textbf{Group 4 (Mixed Strategies):} $0y, 1y, 0\overline{y}, 1\overline{y}$\\

Whenever two different inputs are chosen, there are two possibilities for their values. Either the inputs are complement to each other (i.e., of the form $xy, \overline{xy}$) or they are not complement to each other (i.e., of the form $xy, \overline{x}y$ or $xy, x\overline{y}$). 

Now if a strategy is applied to these chosen inputs, the generated output pair may match in all two positions or only in one position or none of the positions. One can easily explore that for a complement input pair, if the outputs are same then the corresponding strategy must be constant. Similarly if the outputs are complement to each other (i.e., of the form $ab, \overline{ab}$) for a complement input pair, the corresponding strategy must be an input-dependent strategy and if the outputs have only one different bit (i.e., of the form $ab, \overline{a}b$ or $ab, a\overline{b}$) then the corresponding strategy must be a mixed strategy (either from group $3$ or from group $4$). In this similar way, one can also explore the strategies for the cases where the inputs are not complement to each other. For complement input pair and input pair with one bit difference, the different strategies and corresponding outputs are demonstrated in Table \ref{ipopstr}.

\begin{table}[htb]
	\begin{center} 
		\begin{tabular}{|c|c|c|} 
			\hline
			Input & Strategy & Corresponding Output\\
			\hline
			\multirow{3}{*}{Complement input} & Constant & Constant\\
			& Input dependent & Complement output\\
			& Mixed & One bit difference in two outputs\\
			\hline
			Input pair & Constant & Constant\\
			 with one & Input dependent & One bit difference in two outputs\\
			bit difference & Mixed & Same or one bit difference in two outputs\\
			\hline
		\end{tabular}
	\end{center}
	\caption{The strategies and corresponding outputs for different inputs}
	\label{ipopstr}
\end{table}

It is interesting that whenever two different inputs match in exactly one bit position (i.e., inputs of the form $xy, \overline{x}y$ or $xy, x\overline{y}$) but the output bits in that position are different for different inputs then one can't get any strategy that satisfies atleast one output for both the inputs. More formally, whenever the inputs and the corresponding outputs are of the form mentioned in Table \ref{inconsis}, one can't get any strategy that satisfies atleast one output for both the inputs.

\begin{table}[htbp]
	\begin{center}
		\begin{tabular}{ |c|c| } 
			\hline
			Input & Corresponding output\\
			\hline
			\hline
			$xy$ & $ab, a\overline{b}$\\
			\hline
			$x\overline{y}$ & $\overline{a}b, \overline{ab}$\\
			\hline
		\end{tabular}
	\caption{General form of inconsistent outputs for $2+2$ partition}
	\label{inconsis}
	\end{center}
\end{table}


This leads us to the following result.

\begin{theorem}
		\label{thm1}
	{\it For the two input-output pairs of a game, if one bit of the input pair remains the same and the corresponding bit of their outputs is different, then no strategy satisfies atleast one output for both the inputs.}
	
\end{theorem}

\begin{proof}
	: Without loss of generality, here we assume that the input pair is of the form $xy, x\overline{y}$ and the corresponding outputs are of the form $ab$ and $\overline{a}b$ respectively (i.e., the first bit for both the inputs are same however the first bit for the two outputs are different). 
	
	As the two outputs are different, no constant strategy can satisfy both outputs for this input pair. One can also check that whenever a mixed strategy either from group $3$ or from group $4$ is applied to this specified input pair, the first bit of the corresponding outputs always remains the same. However for the given outputs (as specified in Table \ref{inconsis}), the first bits of the outputs for the two different inputs are complement to each other. This implies that no constant or mixed strategy can satisfy atleast one output for both inputs. 
	
	Similarly one can also explore that whenever an input-dependent strategy is applied to this specified input pair, the corresponding outputs are of the form $ab, a\overline{b}$ or $\overline{a}b, \overline{ab}$. This implies that no strategy from any of the groups can satisfy atleast one output for both the inputs. Similarly one can also argue for the other possible input-output pairs of this form. 
\end{proof}
\noindent

From this result, it is clear that if a game has two inputs of the form $xy, x\overline{y} (xy, \overline{x}y)$ and the corresponding outputs are of the form $ab, a\overline{b} (ab, \overline{a}b)$ and $\overline{a}b, \overline{ab} (a\overline{b}, \overline{ab})$ respectively then, there exist no strategy which satisfies atleast one output for both the inputs $xy, x\overline{y} (xy, \overline{x}y)$.

\subsection{Some Basic Results}
\label{bascresult}

In this section, we derive some basic results which are necessary throughout our discussion. It is clear from the group of strategies that for a particular input, the four different strategies of a particular group provide four different outputs. However the two different strategies from two different groups may collide, i. e., may generate the same output for a particular input. For example, the mixed strategy $x0$ (belongs to group 3) and the dependent strategy $xy$ (belongs to group 2) both generate the output $10$ for the input $10$. But the two strategies which provide the same output for a particular input may not provide the same output for any other inputs. For example, the constant strategy $00$ and the dependent strategy $xy$ always provide the same output (i.e., the output $00$) for the input $00$ but these two strategies always provide two different outputs for all the other inputs. Some interesting results (which are required for further analysis) related to these strategies and the groups are mentioned below.\\

\begin{theorem}
	\label{thm3}
	{\it If two different strategies either one from the constant group and the other from dependent group or one from the first mixed group (i.e., group 3) and the other from the second mixed group (i.e., group 4) provide same output for a particular input then, these two strategies must provide two different outputs for all the other inputs.}
\end{theorem} 

\begin{proof}
	:
	Here the two different strategies are either from constant and dependent groups or from the two mixed groups.
	
	\textbf{Case 1:} For every input, there exists a constant and a dependent strategy that provides the same output. Whenever the input changes, the constant strategy always provides the same output as before. However the dependent strategy provides different output than the previous one as the output of a dependent strategy always depends on the inputs and provides different outcomes for different inputs.
	
	\textbf{Case 2:} For every input, there exists a strategy from the first mixed group and another strategy from the second mixed group which provides the same output. In a mixed strategy, one bit of the output is constant and the other bit of the output is input-dependent. So for the first mixed group, there are two types of strategies, namely, $xc_1$ and $\overline{x}c_1$ and for the second mixed group, there are two types of strategies, namely, $c_2y$ and $c_2\overline{y}$ where $c_1$ and $c_2$ denote the constant bits and $x$ and $y$ denote the corresponding dependent bits. This implies that there can be four different choices for the pair of strategies that provide the same output. 
	
	Let us first consider the case where $xc_1$ and $c_2y$ are the two strategies which provide same output for the input $xy$. Then the corresponding outputs are,
	\begin{align*}
	xy &\rightarrow xc_1 ~~~\text{(applying strategy $xc_1$)}\\
	xy &\rightarrow c_2y ~~~\text{(applying strategy $c_2y$)}
	\end{align*}
	
	This two strategies provides same output for this input i.e., $xc_1 = c_2y$.
	
	Now whenever these two strategies $xc_1$ and $c_2y$ are applied to the input $\overline{x}y$, the corresponding outputs are,
	\begin{align*}
	\overline{x}y &\rightarrow \overline{x}c_1 ~~~\text{(applying strategy $xc_1$)}\\
	\overline{x}y &\rightarrow c_2y ~~~\text{(applying strategy $c_2y$)}
	\end{align*}
	
	As $xc_1=c_2y$, $\overline{x}c_1 \neq c_2y$. So the two outputs are different.
	
	Similarly whenever this two strategies are applied to the input $x\overline{y}$, the corresponding outputs are,
	\begin{align*}
	x\overline{y} &\rightarrow xc_1 ~~~\text{(applying strategy $xc_1$)}\\
	x\overline{y} &\rightarrow c_2\overline{y} ~~~\text{(applying strategy $c_2y$)}
	\end{align*}
	
	As $xc_1=c_2y$, $c_2\overline{y} \neq xc_1$. So the two outputs are different.
	
	Similarly whenever this two strategies are applied to the input $\overline{xy}$, the corresponding outputs are,
	\begin{align*}
	\overline{xy} &\rightarrow \overline{x}c_1 ~~~\text{(applying strategy $xc_1$)}\\
	\overline{xy} &\rightarrow c_2\overline{y} ~~~\text{(applying strategy $c_2y$)}
	\end{align*}
	
	As $xc_1=c_2y$, $\overline{x}c_1 = \overline{c_2}y \neq c_2\overline{y}$. So the two outputs are different.
	
	In this similar way, one can also argue the cases for other pair of strategies.
\end{proof}
\noindent

\begin{theorem}
	\label{thm4}
	{\it If a pair of strategies from two distinct groups provide the same output for a particular input $xy$ then this pair of strategies must provide two different outputs for the complement input $\overline{xy}$.}
\end{theorem}	

\begin{proof}
	: Here $xy$ is the input for which two different strategies from two different groups provide the same output. From the result of theorem \ref{thm3}, it can be easily argued that if the two strategies are from constant and dependent groups or from the two mixed groups then, these two strategies must provide two different outputs for the complement input $\overline{xy}$. 
	
	So there are two remaining cases that may occur. The first case is that whenever one strategy is from the constant group and the other strategy is from any one of the two mixed groups and the second case is that whenever one strategy is from the dependent group and the other strategy is from any one of the two mixed groups.
	
	{\bf Case 1:} In this case, one strategy from the constant group and the other strategy from one of the two mixed groups provide the same output (say $ab$) for the input $xy$. Whenever these two strategies are applied to the complement input $\overline{xy}$, then one can easily check that the constant strategy provides the output $ab$ but, the mixed strategy provides the output either $\overline{a}b$ or $a\overline{b}$.
	
	{\bf Case 2:} Similarly in this case, if one strategy from the dependent group and the other strategy is from one of the two mixed groups provide the same output (say $ab$) for the input $xy$, then the dependent strategy provides the output $\overline{ab}$ and the mixed strategy provides the output either $\overline{a}b$ or $a\overline{b}$ for the complement input $\overline{xy}$. This proves the result.
\end{proof}
\noindent

\begin{corollary}
	\label{cor1}
	{\it From the results of theorem \ref{thm3} and theorem \ref{thm4}, one can conclude that whenever there are two strategies in which one is from the constant (dependent) group and the other is from any one of the two mixed groups provide the same output for the input $xy$, then these strategies may not always provide two different outputs for the input $x\overline{y}$ and $\overline{x}y$.}
\end{corollary} 

For example, the dependent strategy $x\overline{y}$ and the mixed strategy $0\overline{y}$ both provide the output $00$ for the input $01$ however these two strategies also provide the output $01$ for the input $00$. So whenever the inputs are not complement to each other, one can't conclude anything about the outcomes. 

\begin{theorem}
	\label{thm2}
	{\it For a complement input pair (i.e., for two inputs of the form $xy$ and $\overline{xy}$), if one input has $m$ many outputs and the other input has $n$ many outputs then there are exactly $mn$ many strategies such that each of them satisfies an output for both the inputs.}
\end{theorem}
\proof :
For a complement input pair, we show that if each
input has exactly one valid output, then there is exactly one
strategy that satisfies both inputs. Let us consider that the input $xy$ has output $ab$ and input $\overline{xy}$ has any one of the four outcomes $ab, \overline{a}b, a\overline{b}$ and $\overline{ab}$.

One can easily check that whenever $\overline{xy}$ has output $ab$ then the common strategy is a constant strategy, whenever $\overline{xy}$ has output either $\overline{a}b$ or $a\overline{b}$ then the common strategy is a mixed strategy and whenever $\overline{xy}$ has output $\overline{ab}$ then the common strategy is an input-dependent strategy which satisfies the outputs for both the inputs $xy$ and $\overline{xy}$.

This implies that there must be a strategy corresponding to every different output pair for two complement inputs. So for a complement input pair, if one input has $m$ many outcomes and the other input has $n$ many outcomes, then there are exactly $mn$ different pair of outcomes. Moreover for each pair of outcomes, there exist a strategy that satisfies the outcomes for both the inputs. So, there are exactly $mn$ many strategies which satisfy an output for both the inputs.\qed

\noindent

\begin{theorem}
	\label{thm5}
	{\it If an input (say $xy$) has a complement output pair and the corresponding complement input (i.e., $\overline{xy}$) also has two outcomes (may not be complement), then the four strategies corresponding to this input pair $xy, \overline{xy}$ must provide four different outcomes for atleast one of the rest two inputs (i.e., for inputs $\overline{x}y$ and $x\overline{y}$).}
\end{theorem}
\proof :
Whenever each of the inputs of the complement input pair $xy, \overline{xy}$ has two outcomes and the input $xy$ has a complement output pair, the input $\overline{xy}$ has two possibilities for output. Either the outcomes of $\overline{xy}$ are complement to each other or they are not complement to each other.

\textbf{Case 1:} Whenever the input $\overline{xy}$ has complement output pair, then also there are two possibilities. Either $\overline{xy}$ has the same complement pair as in $xy$ or the complement pair of $\overline{xy}$ is different from the output of $xy$.

Whenever $xy$ and $\overline{xy}$ have the same complement output pair, one can check that the common strategies are $2$ constant and $2$ dependent strategies. We can easily verify that for the input pairs $xy, \overline{xy}$ whenever a constant and a dependent strategy collide for a particular input, the same constant
and dependent strategy must not collide for the other input (rather the same constant strategy collide with the other dependent strategy for the other input). From the result of theorem \ref{thm3}, we can argue that for this case, the four common strategies must provide four different outputs for all the rest two inputs. Similarly one can conclude this same result for the case when $\overline{xy}$ has a different complement output pair than $xy$.

\textbf{Case 2:} Whenever the input $\overline{xy}$ has non-complement output pair, then the common strategies are one constant, one input-dependent and two mixed strategies for $xy, \overline{xy}$ pair. One can verify that among two mixed strategies, one collides with the constant strategy and the other collides with the dependent strategy for input $xy$. But for input $\overline{xy}$, the two mixed strategies and the constant and the dependent strategy collide among themselves. Now for any one of the remaining two inputs, the constant strategy collides with
those mixed strategies for which they provide different outputs
for the input $xy$ and similarly for the dependent and other mixed
strategies. So from the result of theorem \ref{thm3} and theorem \ref{thm4}, one can conclude that these four strategies must provide four different outputs for the remaining input. This concludes the proof.\qed

\textbf{Note}: {\it If both the inputs $xy$ and $\overline{xy}$ have non-complement output pair, then the four different strategies may not provide four different outputs for any of the remaining two inputs $\overline{x}y$ and $x\overline{y}$}.

\begin{lemma}
\label{lem1}	
{\it : If an input (say $xy$) has a complement output pair and its complement input (i.e., $\overline{xy}$) has only one outcome, then the two strategies from $xy, \overline{xy}$ pair must provide non-complement output pairs for each of the rest two inputs.}
\end{lemma}

\begin{proof}
: Let us consider that the input $xy$ has two complement outputs of the form $ab$ and $\overline{ab}$. Then the complement input $\overline{xy}$ has two possibilities, either the output of $\overline{xy}$ is same as one of the outputs of $xy$ or the output of $\overline{xy}$ is different from the output of $xy$.

\textbf{Case 1:} Whenever the output of $\overline{xy}$ is same with one of the output of $xy$, there are one constant strategy and one input-dependent strategy. Let us assume that the common output of $xy$ and $\overline{xy}$ is $ab$. Then the constant strategy is $ab$ and the input-dependent strategy is $mn$ (say) where $m \in \{x, \overline{x}\}$ and $n \in \{y, \overline{y}\}$. This implies that,
\begin{align*}
xy &\rightarrow \overline{ab} ~~~\text{(applying strategy $mn$)}\\
\overline{xy} &\rightarrow ab ~~~\text{(applying strategy $mn$)}
\end{align*}
Thus we conclude that
\begin{align*}
x &\rightarrow \overline{a}~~~\text{(applying strategy  $m$)}\\
y &\rightarrow \overline{b}~~~\text{(applying strategy $n$)}\\
\end{align*}

Hence it is clear that the strategy
$mn$ provides the output $a\overline{b}$ for the input $\overline{x}y$ and provides the
output $\overline{a}b$ for input $x\overline{y}$. So the two strategies $ab$ and $mn$
provide non-complement output pair $ab$ and $a\overline{b}$ for the input $\overline{x}y$ and provide non-complement output pair $ab$ and $\overline{a}b$ for the input $x\overline{y}$.

\textbf{Case 2:} In a similar way, whenever the output of $\overline{xy}$ is different from the outputs of $xy$, there are two mixed strategies from two different groups. Let us consider that the outputs of $xy$ are $ab$, $\overline{ab}$ and the output of $\overline{xy}$ is either $\overline{a}b$ or $a\overline{b}$. Without loss of generality, we assume that the output of $\overline{xy}$ is $\overline{a}b$. Let the two mixed strategies are $mc_1$ and $c_2n$, where $m \in \{x, \overline{x}\}$ and $n \in \{y, \overline{y}\}$ are the dependent bits and $c_1, c_2 \in \{0, 1\}$ are the constant bits. If we consider that the strategy $mc_1$ provides output $ab$ and $c_2n$ provides output $\overline{ab}$ for input $xy$, then one can easily check that,
\begin{align*}
x &\rightarrow a~~~\text{(applying strategy $m$)}\\
y &\rightarrow \overline{b}~~~\text{(applying strategy $n$)}
\end{align*}

It is also clear that $c_1 = b$ and $c_2 = \overline{a}$.

Hence, one can easily check that the strategy $mc_1$ and $c_2n$ provide outputs $\overline{a}b$ and $\overline{a}\overline{b}$ respectively for input $\overline{x}y$. Similarly one can also check that the strategy $mc_1$ and $c_2n$ provide outputs $ab$ and $\overline{a}b$ respectively for input $x\overline{y}$. This implies that the mixed strategies from two different groups also provide non-complement output pair for the rest of the two inputs. Similarly one can also check the other cases.
\end{proof}

\subsection{Analysis of the Maximum Success Probability in Classical Scenario}
\label{class}

In the classical scenario of a nonlocal game, the players have to fix some strategies before the game begins. After getting the input bits from the referee, the players aren't allowed to communicate with each other. For the binary input binary output two-party nonlocal games, each player has only two possibilities for their input bits (either $0$ or $1$) and has two choices (either $0$ or $1$) for the output bits. In this scenario, each player has atmost $4$ different strategies (either output $0$ or $1$ or the input bit itself or the complement of the input bit) to generate their output bits. This implies that for a particular two-bit input string provided by the referee, there are atmost $16$ different strategies for the two players in a classical scenario. 

Among these $16$ different strategies, the strategy which provides the maximum success probability for a particular game is the optimal classical strategy and the corresponding success probability is the maximum classical success probability for this game. All the possible output strategies (by the two players) for a particular game and their corresponding success probabilities can be represented in a tabular form as mentioned in Table \ref{classical}.

In this Table, each $p_i$ denotes the fraction of inputs for which the game can be won using the $i$-th strategy. Two players can choose any of these $16$ different strategies before the game begins and later can output their bits accordingly. For example, if they choose the first strategy specified in Table \ref{classical}, both of them can choose $0$ as an output irrespective of their inputs. Similarly whenever they choose the second strategy, the first player always outputs $0$ irrespective of his inputs whereas, the second player outputs his corresponding input bit itself, i.e., if he receives the input $0$, he outputs $0$, otherwise he outputs $1$. After their output, the referee checks the fraction of inputs for which the players win the game. The strategy which generates the winning outcomes for most of the inputs of a particular game is considered as the optimal strategy corresponding to that game. This implies that from Table \ref{classical}, one can obtain the maximum classical success probability ($p_{max}$) as $p_{max} = \max_{i} p_i$.

\begin{table}[hb]
	\begin{center}
		\setlength{\tabcolsep}{0.35em}\scalebox{1}{
			\begin{tabular}{ |c|c|c|c|c| } 
				\hline
				\multicolumn{2}{|c|}{\bf Output for Alice ($a$)} & \multicolumn{2}{|c|}{\bf Output for Bob ($b$)} & ~\\
				\hline
				Output for & Output for & Output for & Output for & Success \\
				input $x=0$ & input $x=1$ & input $y=0$ & input $y=1$ & Probability  \\
				\hline
				0 & 0 & 0 & 0 & $p_1$\\
				0 & 0 & 0 & 1 & $p_2$\\
				0 & 0 & 1 & 0 & $p_3$\\
				0 & 0 & 1 & 1 & $p_4$\\
				0 & 1 & 0 & 0 & $p_5$\\
				0 & 1 & 0 & 1 & $p_6$\\
				0 & 1 & 1 & 0 & $p_7$\\
				0 & 1 & 1 & 1 & $p_8$\\
				1 & 0 & 0 & 0 & $p_9$\\
				1 & 0 & 0 & 1 & $p_{10}$\\
				1 & 0 & 1 & 0 & $p_{11}$\\
				1 & 0 & 1 & 1 & $p_{12}$\\
				1 & 1 & 0 & 0 & $p_{13}$\\
				1 & 1 & 0 & 1 & $p_{14}$\\
				1 & 1 & 1 & 0 & $p_{15}$\\
				1 & 1 & 1 & 1 & $p_{16}$\\
				\hline
		\end{tabular}}
	\end{center}
	\caption{Success probabilities of a game for all possible classical strategies} 
	\label{classical}
\end{table}

As every binary input binary output two-party nonlocal game has $4$ possible inputs, one can easily check that the classical success probability for each of the possible $16$ strategies must belong to the set $\{0, 0.25, 0.5, 0.75, 1\}$. From the result of theorem~\ref{thm2}, it is clear that there must be a classical strategy corresponding to each of the complement input pairs. This implies that for every binary input binary output two-party nonlocal game, the maximum classical success probability must be atleast $0.5$. Similarly if a game has inconsistent outputs (as mentioned in subsection~\ref{incos}) for an input pair, then according to the discussion of subsection~\ref{incos}, the maximum classical success probability must be less than $1$ (i.e., either $0.5$ or $0.75$). 

\subsection{Analysis of the Maximum Success Probability in Quantum Scenario}
\label{quant}

In the quantum strategy of a two-party nonlocal game, the two players initially share some entanglement among themselves (before the game begins) and then during the game, they perform some specific (unitary) operations on their qubits (based on the inputs) and measure their qubits to get the output bits.

Let us assume that the two players (say Alice and Bob) share the Bell-state $\vert \psi \rangle _{AB} = \frac{\vert 00 \rangle + \vert 11 \rangle }{\sqrt{2}}$ among themselves and Alice measures in $\theta_0 (\theta_1)$ rotated basis for the input $0 (1)$ and Bob measures in $\psi_0 (\psi_1)$ rotated basis for the input $0 (1)$. 

Now, it is easy to check that whenever the referee provides the input bit $0$ to both Alice and Bob (i.e., for the input $00$), the shared states between Alice and Bob (after applying their respective unitary operators) are of the form

\begin{eqnarray*}
	&&\frac{1}{\sqrt{2}}\left[(\cos{\theta_0}\vert0\rangle + \sin{\theta_0}\vert1\rangle)(\cos{\psi_0}\vert0\rangle + \sin{\psi_0}\vert1\rangle)\right] \\
	&~&+\frac{1}{\sqrt{2}}\left[(-\sin{\theta_0}\vert0\rangle+\cos{\theta_0}\vert1\rangle)(-\sin{\psi_0}\vert0\rangle+\cos{\psi_0}\vert1\rangle)\right]\\
	&=&\frac{1}{\sqrt{2}}\left[(\cos{\theta_0}\cos{\psi_0}+\sin{\theta_0}\sin{\psi_0})\vert 00 \rangle  + (\cos{\theta_0}\sin{\psi_0}-\sin{\theta_0}\cos{\psi_0})\vert 01 \rangle \right]\\
	&~&+\frac{1}{\sqrt{2}}\left[(\sin{\theta_0}\cos{\psi_0}-\cos{\theta_0}\sin{\psi_0})\vert 10 \rangle  + (\cos{\theta_0}\cos{\psi_0}+\sin{\theta_0}\sin{\psi_0})\vert 11 \rangle \right]\\
	&=& \frac{1}{\sqrt{2}} \left[\cos(\theta_0-\psi_0)\vert 00 \rangle  - \sin(\theta_0-\psi_0)\vert 01 \rangle  + \sin(\theta_0-\psi_0)\vert 10 \rangle  + \cos(\theta_0-\psi_0)\vert 11 \rangle \right]
\end{eqnarray*}

So for the input $00$, the probability of getting each of the outputs $00$ and $11$ is $\frac{1}{2}\cos^2(\theta_0-\psi_0)$ and the probability of getting each of the outputs $01$ and $10$ is $\frac{1}{2}\sin^2(\theta_0-\psi_0)$.

Similarly for the input $01$, the shared states between Alice and Bob after applying the specific unitary operations are of the form $\cos(\theta_0-\psi_1)\vert 00 \rangle  - \sin(\theta_0-\psi_1)\vert 01 \rangle  + \sin(\theta_0-\psi_1)\vert 10 \rangle  + \cos(\theta_0-\psi_1)\vert 11 \rangle $. So in this case, the probability of getting each of the outputs $00$ and $11$ is $\frac{1}{2}\cos^2(\theta_0-\psi_1)$ and the probability of getting each of the outputs $01$ and $10$ is $\frac{1}{2}\sin^2(\theta_0-\psi_1)$.

In this similar way, one can easily check that for the input $10$, the shared states between Alice and Bob after applying the specific unitaries are of the form $\cos(\theta_1-\psi_0)\vert 00 \rangle  - \sin(\theta_1-\psi_0)\vert 01 \rangle  + \sin(\theta_1-\psi_0)\vert 10 \rangle  + \cos(\theta_1-\psi_0)\vert 11 \rangle $ and the corresponding probabilities are $\frac{1}{2}\cos^2(\theta_1-\psi_0)$ (for each of the outputs $00$ and $11$) and $\frac{1}{2}\sin^2(\theta_1-\psi_0)$ (for each of the outputs $01$ and $10$).

Similarly for the input $11$, the shared states between Alice and Bob after applying the specific unitaries are of the form $\cos(\theta_1-\psi_1)\vert 00 \rangle  - \sin(\theta_1-\psi_1)\vert 01 \rangle  + \sin(\theta_1-\psi_1)\vert 10 \rangle  + \cos(\theta_1-\psi_1)\vert 11 \rangle $ and the corresponding probabilities are $\frac{1}{2}\cos^2(\theta_1-\psi_1)$ (for each of the outputs $00$ and $11$) and $\frac{1}{2}\sin^2(\theta_1-\psi_1)$ (for each of the outputs $01$ and $10$).

From these quantum success probability expressions for different inputs, it is clear that for a particular input $xy$, the probability of getting each of the outputs $00$ and $11$ is $\frac{1}{2}\cos^2\alpha$ and the probability of getting each of the outputs $01$ and $10$ is $\frac{1}{2}\sin^2\alpha$ where $\alpha = (\theta_x-\psi_y)$ (according to our mentioned strategy).

As for every nonlocal game, the referee is supposed to provide the input bits randomly, the two players calculate the overall success probability considering each of the inputs equally likely. For any particular game, the expression of the quantum success probability depends on the distribution of the successful outcomes for all the possible inputs. Depending on this distribution, the quantum success probability expressions involve the variables $\theta_0, \theta_1, \psi_0, \psi_1$. After getting the quantum success probability expression for a particular game, one can easily find the values of $\theta_0, \theta_1, \psi_0, \psi_1$ for which the success probability becomes maximum. For demonstration, here we can consider the example of the CHSH Game that corresponds to the partition $2+2+2+2$. From Subsection~\ref{parchsh}, one can conclude that the quantum success probability (according to the discussion of this section) of the CHSH game is $\frac{1}{4} \left[\cos^2(\theta_0 - \psi_0) + \cos^2(\theta_0 - \psi_1) + \cos^2(\theta_1 - \psi_0) + \sin^2(\theta_1 - \psi_1)\right]$. Now, one can vary different values of $\theta_0, \theta_1, \psi_0, \psi_1$ to get the corresponding maximum quantum success probability. From the analysis in Subsection~\ref{parchsh} (and also from the previously known results), it can be concluded that one can find some values of $\theta_0, \theta_1, \psi_0, \psi_1$ for which the quantum success probability will be $0.853$, which is maximum for this partition (for a detailed analysis of this calculation, one may refer to the Subsection~\ref{parchsh}). Similarly for the other partition of games, one can calculate the maximum quantum success probability by following the strategy discussed in this subsection.  

One can verify that for the games having inconsistent outputs, the quantum success probabilities corresponding to the inconsistent input pair are of the form $\frac{1}{2}\cos^2\alpha$ and $\frac{1}{2}\sin^2\alpha$. So, the maximum quantum success probability for these games having inconsistent outputs is $0.75$. This implies that all the nonlocal games having inconsistent outputs may not offer a quantum advantage (i.e., the maximum quantum success probability is greater than the maximum classical success probability)\footnote{In this context one should remember that every classical strategy is also a quantum one where no entanglement is shared between the parties}.

\section{Analysis of the Results in \cite{NGA19}}
\label{error}

It is evident from the discussion till now that every binary input binary output two-party nonlocal game can be represented as a $4$-variable boolean function. One can also consider the inputs and the outputs separately as $2$-variable boolean functions for binary input binary output two-party nonlocal games and compose these two functions to construct $4$-variable functions. For example in the CHSH game, the input function is $f(x,y)=x \wedge y$ and the output function is $g(a,b)=a\oplus b$. The actual function that represents the CHSH game is just the composition of these two (input and output) functions.

Recently some analysis has been done in this direction in~\cite{NGA19} (considering all the non-constant $2$-variable boolean functions and composing every possible pairs among them to construct the corresponding $4$-variable boolean functions) to explore the performance of some $4$-variable boolean functions (or binary input binary output two-party nonlocal games) as distinguishers for the certification of different dimensional quantum states. As the authors consider only non-constant boolean functions in~\cite{NGA19}, the total number of $4$-variable boolean functions that they have considered are $(2^{2^2}-2)\times(2^{2^2}-2) = 14\times14 = 196$. However, there are total $2^{2^4}=65536$ possible $4$-variable boolean functions. So in~\cite{NGA19}, only a small fraction of the games are explored from the class of all possible binary input binary output two-party nonlocal games.

There are some miscalculations in~\cite[Table 1]{NGA19} regarding the number of different boolean functions which we would like to point out here. In Table $1$ of~\cite{NGA19}, it is mentioned that the total number of function pairs $(f,g_2)$ such that each of $f(x,y)$ and $g_2(a,b)$ contains one $0 (1)$ is $32 (32)$. However one can verify that the total number of such function pairs is actually $16$. For a detail analysis corresponding to this result, one may refer to Proposition~\ref{totcount}.

\begin{proposition}
	\label{totcount}
	Let $g_i: \mathbb{Z}_2 \times \mathbb{Z}_2 \rightarrow \mathbb{Z}_2$ be a boolean function such that $\mid g_i^{-1}(0) \mid = 1$ where $g_i^{-1}(0) = \{ (x, y) \in \mathbb{Z}_2 \times \mathbb{Z}_2: g_i(x, y) = 0 \}$ for $i= 1, 2$. Let $f(x_1, x_2, x_3, x_4) = g_1 (x_1, x_2) * g_2(x_3, x_4)$ where $*$ is a binary operation on $\mathbb{Z}_2$. Then given a binary operation $*$, there are atmost $16$ different possibilities for $f$.   
\end{proposition}

\begin{proof}
	:~Since $\mid g_i^{-1}(0) \mid = 1$, there is $(a_i, b_i)$ such that $g_i(a_i, b_i) = 0$ and $g_i(x, y) = 1$ for $(x, y) \neq (a_i, b_i)$. Now, for each of the two functions $g_1$ and $g_2$, there are four different choices that satisfy the above condition namely, 
	\begin{align}
		g_1^{(1)}(0, 0) &= 0 \text{ and } g_1^{(1)} (x, y) = 1 \text{ for } (x, y) \neq (0, 0)\\
		g_1^{(2)}(0, 1) &= 0 \text{ and } g_1^{(2)} (x, y) = 1 \text{ for } (x, y) \neq (0, 1)\\
		g_1^{(3)}(1, 0) &= 0 \text{ and } g_1^{(3)} (x, y) = 1 \text{ for } (x, y) \neq (1, 0)\\
		g_1^{(4)}(1, 1) &= 0 \text{ and } g_1^{(4)} (x, y) = 1 \text{ for } (x, y) \neq (1, 1). 
	\end{align} 
	Similarly for $g_2$, there are also four different choices namely, 
	\begin{align}
		g_2^{(1)}(0, 0) &= 0 \text{ and } g_2^{(1)} (x, y) = 1 \text{ for } (x, y) \neq (0, 0)\\
		g_2^{(2)}(0, 1) &= 0 \text{ and } g_2^{(2)} (x, y) = 1 \text{ for } (x, y) \neq (0, 1)\\
		g_2^{(3)}(1, 0) &= 0 \text{ and } g_2^{(3)} (x, y) = 1 \text{ for } (x, y) \neq (1, 0)\\
		g_2^{(4)}(1, 1) &= 0 \text{ and } g_2^{(4)} (x, y) = 1 \text{ for } (x, y) \neq (1, 1). 
	\end{align} 
	Thus $f(x_1, x_2, x_3, x_4) = g_1^{(i)} (x_1, x_2) * g_2^{(j)} (x_3, x_4)$ for some $1 \leq i, j \leq 4$. As there are maximum $4$ different choices for each of $g_1^{(i)}$ and $g_2^{(j)}$ for some $1 \leq i, j \leq 4$, there are atmost $4 \times 4 = 16$ different choices for $f$.
\end{proof}

It is also mentioned (in~\cite{NGA19} Table $1$) that the total number of function pairs such that $f(x,y)$ contains one $0$ and $g_2(a,b)$ contains one $1$ are $6$. For this case also, similar to the derivation performed in the proof of proposition~\ref{totcount}, one can verify that the total number of such function pairs is also $16$.

In~\cite{NGA19}, the authors have proposed the idea of distinguishing different dimensional quantum states with the help of some nonlocal games. In their paper, they have explored the performance of some two-party nonlocal games in quantum scenario with the intention of finding those games which provide significant advantage in the quantum winning probability for different dimensional states. However, the main limitation in their approach is that they have explored only $196$ functions from the set of $65536$ possible $4$-variable Boolean functions. Because of this limitation, they might not consider the game which is the most efficient as the dimensionality distinguisher (i.e., which has the maximum probability difference in the quantum scenario for different dimensional states) among all possible binary input binary output two-party nonlocal games. In this article, we have considered all those binary input binary output two-party nonlocal games which have atleast one successful outcome for every possible input and evaluate their performance both in classical and quantum scenario (considering the strategy mentioned in Subsection~\ref{quant}). However, we haven't analyzed anything regarding the performance of those games as dimensionality distinguishers.

\section{Analysis of the Binary Input Binary Output Two-party Nonlocal Games}
\label{main}

In the current context, we are interested in finding all those two-party binary input binary output nonlocal games where one can achieve quantum advantage. So far, CHSH game is the most well known game that offers a separation around $0.1$ between the maximum classical (which is $0.75$) and the maximum quantum (which is around $0.853$) success probability.

From the definition of the partition introduced in definition~\ref{defpart}, one can easily check that the CHSH game can be represented as a $2+2+2+2$ partition based on the distribution of its outputs. As our main intention is to find out all those games that offer quantum advantage (with maximum quantum success probability greater than the existing maximum for the two party scenario, i.e., $0.853$), here we consider only those games for which the number of valid outputs corresponding to each possible input is non-zero (so that there is a chance of achieving the maximum quantum success probability greater than $0.853$ for random inputs).

For every number of successful outcomes (i.e., the number of $0'$s in the output column of a boolean function), we first find out all possible partitions of that outcome and then explore the performance of the games corresponding to each of those partitions to derive the maximum classical and the maximum quantum success probabilities. For example, the games having $8$ successful outcomes (i.e., $8$ number of $0'$s in the output column of the boolean function representation of those games), there are four possible partitions such that each input has atleast one successful outcome. In this section, we first find out all those partitions for every possible number of successful outcomes and then analyze the performance of the games corresponding to each of those partitions according to the techniques mentioned in Subsection \ref{class} and Subsection \ref{quant}. 

\subsection{Games Corresponding to $8$ Successful Outcomes}

In this subsection, we analyze (in details) the performance of the games corresponding to all possible partitions for $8$ successful outcomes in both classical and quantum scenario.

\subsubsection{\it Analysis for partition $4+2+1+1$:}

For this partition of games, there must be a complement input pair (i.e., of the form $xy$ and $\overline{xy}$) either both the inputs have $1$ outcome or one input has $1$ outcome and the other input has $2$ outcomes. Whenever the two $1$ outcomes have a complement input pair, the strategy corresponding to this input pair must satisfy one output for the input corresponding to $4$ outcomes. Similarly if the $2$ outcome and a $1$ outcome has a complement input pair, then the two strategies corresponding to this complement input pair must satisfy atleast one output for the input having $4$ outcomes.

So the minimum classical success probability for all the games of this partition is $0.75$. From the discussion in subsection \ref{quant}, one can easily check that the maximum quantum success probability corresponding to each of the inputs having $4$ and $2$ outcomes is $1$ and for inputs having $1$ outcome is $0.5$. So, the maximum quantum success probability for any game of this partition is  $\frac{1}{4}[1+1+0.5+0.5] = 0.75$. This implies that for this partition of games, one can't achieve any advantage in quantum scenario. For example, here we consider the following game corresponding to this partition.

\begin{table}[htbp]
	\begin{center}
		\begin{tabular}{ |c|c| } 
			\hline
			Input & Corresponding output\\
			\hline
			\hline
			00 & 00, 01, 10, 11\\
			\hline
			01 & 00, 11\\
			\hline
			10 & 01\\
			\hline 
			11 & 10\\
			\hline
		\end{tabular}
	\end{center}
	\label{4211}
\end{table}

For this game, one can easily check that for the strategy $a=x$ and $b=0$, the players can win the game with probability $0.75$ in classical scenario whereas in quantum scenario, the maximum quantum success probability is $0.75$. This implies that for this partition of games, there is no chance of getting a quantum advantage.\\	

\subsubsection{\it Analysis for partition $3+3+1+1$}

For this partition of games, there must be a complement input pair either they have $3$ outcomes or one input has $3$ and the other has $1$ outcome. Let us consider that the input $xy$ has $3$ outcomes. Then for the complement input $\overline{xy}$, there are two possibilities, either $\overline{xy}$ has $3$ outcomes or $\overline{xy}$ has $1$ outcome.

\textbf{Case 1:} Whenever $\overline{xy}$ has $3$ outcomes, one can get nine strategies for $xy, \overline{xy}$ pair. As each of the inputs $xy$ and $\overline{xy}$ must have a complement output pair, according to the result of theorem \ref{thm5}, these four strategies corresponding to these two complement output pairs must provide four different outputs for each of the rest two inputs. So one of these four strategies must satisfy atleast one output for atleast one of the rest of two inputs. Hence, the minimum classical success probability for all these games is $0.75$. This implies that for this partition of games, one can't achieve any advantage in quantum scenario. For example, here we consider the following game corresponding to this partition.

\begin{table}[htbp]
	\begin{center}
		\begin{tabular}{ |c|c| } 
			\hline
			Input & Corresponding output\\
			\hline
			\hline
			00 & 00, 10, 11\\
			\hline
			01 & 00\\
			\hline
			10 & 11\\
			\hline 
			11 & 00, 10, 11\\
			\hline
		\end{tabular}
	\end{center}
	\label{3311a}
\end{table}

For this game, one can easily check that for the strategy $a=0, b=0$ or $a=1, b=1$ or $a=x, b=\overline{y}$, the players can win the game with probability $0.75$ in classical scenario whereas in quantum scenario, the maximum quantum success probability is $0.75$. This implies that for this partition of games, there is no chance of getting a quantum advantage.

\textbf{Case 2:} Similarly whenever $\overline{xy}$ has $1$ valid outcome, for $xy,\overline{x} \overline{y}$ pair one can get three strategies. Among these three strategies, the two strategies that correspond to the complement output pair of $xy$ must provide $2$ different outputs for each of the rest two inputs (as mentioned in theorem \ref{thm5}). So one of these two strategies must satisfy atleast one output for the rest input having $3$ outcomes. Hence, the minimum classical success probability for this form of game is also $0.75$. This implies that for this partition of games, one can't achieve any advantage in quantum scenario. For example, here we consider the following game corresponding to this partition.

\begin{table}[htbp]
	\begin{center}
		\begin{tabular}{ |c|c| } 
			\hline
			Input & Corresponding output\\
			\hline
			\hline
			00 & 00, 01, 11\\
			\hline
			01 & 00, 10, 11\\
			\hline
			10 & 01\\
			\hline 
			11 & 11\\
			\hline
		\end{tabular}
	\end{center}
	\label{3311b}
\end{table}

For this game, one can easily check that for the strategy $a=1, b=1$ or $a=\overline{x}, b=\overline{y}$, the players can win the game with probability $0.75$ in classical scenario. 

From the discussion of subsection \ref{quant}, one can easily verify that the maximum quantum success probability for each of the inputs having $3$ outcomes is $1$ and for each of the inputs having $1$ outcome is $0.5$. So, the maximum quantum success probability for any game of this partition (for equiprobable outcomes) is  $\frac{1}{4}[1+1+0.5+0.5] = 0.75$.

It is clear from the analysis that for this partition of games, there is no chance of getting any advantage in quantum success probability as compared to the classical one.\\

\subsubsection{\it A Game for partition $2+2+2+2$ having quantum advantage}
\label{parchsh}

For this partition of games, each of the four inputs has $2$ outcomes. According to the discussion of subsections \ref{class} and \ref{quant}, one can easily check that if none of the inputs have complement output pair, the maximum quantum success probability is $0.5$. Whenever $1$ or $2$ inputs have complement output pair (like the discussions of the previous two partitions), one can easily check that there is no advantage in quantum success probability. So to achieve quantum advantage, the games must have complement output for atleast three inputs. For the games having complement output for three inputs, one can verify that although the maximum quantum success probability is greater than $0.75$, the maximum classical success probability is always $1$. A well-known game of this partition having complement output pair for all the inputs is the CHSH game which offers quantum advantage. Here we consider this game and analyze its performance in both classical and quantum scenarios.

\begin{table}[htbp]
	\begin{center}
		\begin{tabular}{ |c|c| } 
			\hline
			Input & Corresponding output\\
			\hline
			\hline
			00 & 00, 11\\
			\hline
			01 & 00, 11\\
			\hline
			10 & 00, 11\\
			\hline 
			11 & 01, 10\\
			\hline
		\end{tabular}
	\end{center}
	\label{2222}
\end{table}


From the strategies mentioned in subsection \ref{class}, one can easily check that the maximum classical success probability for this game is $0.75$ and one of the strategies to get this success probability is $a=0$ and $b=0$.

Similarly from the discussion of subsection \ref{quant}, one can easily check that the expression for quantum success probability of this game is of the form

\begin{align*}
&~ \frac{1}{4} \left[\cos^2(\theta_0 - \psi_0) + \cos^2(\theta_0 - \psi_1) + \cos^2(\theta_1 - \psi_0) + \sin^2(\theta_1 - \psi_1)\right]\\
&= \frac{1}{2} + \frac{1}{8} \left[\cos{2\alpha} + \cos{2\beta} + \cos{2\gamma} - \cos{2\delta}\right]
\end{align*}
where $\alpha=(\theta_0 - \psi_0)$, $\beta=(\theta_0 - \psi_1)$, $\gamma=(\theta_1 - \psi_0)$ and 
$\delta=(\theta_1 - \psi_1)$. 

Now the cosines can be written as an inner product between two unit vectors. 
Suppose,
\begin{align*}
&~ u_{0} = \cos{\theta_0} \vert 0 \rangle + \sin {\theta_0} \vert 1 \rangle\\
&~ u_{1} = \cos {\theta_1} \vert 0 \rangle + \sin {\theta_1} \vert 1 \rangle\\
&~ v_{0} = \cos{\psi_0} \vert 0 \rangle + \sin {\psi_0} \vert 1 \rangle\\
&~ v_{1} = \cos{\psi_1} \vert 0 \rangle+ \sin {\psi_1} \vert 1 \rangle
\end{align*}

Then one can easily check that for all $i,j$, $u_{i}v_{j} = \cos{2(\theta_{i} - \psi_{j})}$. So one can rewrite the above expression as

\begin{align*}
&~ \frac{1}{2} + \frac{1}{8} [ u_0v_0 + u_0v_1 + u_1v_0 - u_1v_1]\\
&= \frac{1}{2} + \frac{1}{8} [ u_0 (v_0 + v_1) + u_1 (v_0 - v_1)]\\
&\leq \frac{1}{2} + \frac{1}{8} ( ||v_0 + v_1|| + ||v_0 - v_1||)
\end{align*}

Let us assume, $\langle v_0, v_1 \rangle = a + ib$ and $\langle v_1, v_0 \rangle = a - ib$. Then $||v_0 + v_1|| = \sqrt{2 + 2a}$ and $||v_0 - v_1|| = \sqrt{2 - 2a}$. It is easy to check that the expression $\sqrt{2 + 2a} + \sqrt{2 - 2a}$ attains maximum value for $a = 0$ and the corresponding maximum value is $2\sqrt{2}$ i.e., $ ( ||v_0 + v_1|| + ||v_0 - v_1||) \leq 2\sqrt{2}$. So, the maximum quantum success probability for this form of games is $\left(\frac{1}{2} + \frac{2\sqrt{2}}{8}\right) = \frac{1}{2} + \frac{1}{2\sqrt{2}} \approx 0.853$.

This implies that one can find some values of $\theta_0, \theta_1, \psi_0$ and $\psi_1$ for which the game can be won with probability $0.853$ in quantum scenario. Similarly one can show the same upper bound for some other games of this group for which the maximum classical success probability is $0.75$. Therefore, the maximum separation for this partition of games is $(0.853-0.75) \approx 0.103$.\\

\subsubsection{\it A Game for partition  $3+2+2+1$ having quantum advantage}

From the theoretical analysis (similar to the analysis of partition $4+2+1+1$ and $3+3+1+1$), one can easily check that not all games for this partition can be won with probability $1$ in classical scenario and there are some games for which the maximum classical success probability is $0.75$. From the expressions of quantum success probabilities of these games, one can easily check that for some of those games, maximum quantum success probability is greater than the classical one. Here we consider one of these games and analyze its performance in both classical and quantum scenarios.

\begin{table}[htbp]
	\begin{center}
		\setlength{\tabcolsep}{0.35em}\scalebox{1}{
			\begin{tabular}{ |c|c| } 
				\hline
				Input & Corresponding output\\
				\hline
				\hline
				00 & 00, 01, 11\\
				\hline
				01 & 00, 11\\
				\hline
				10 & 01\\
				\hline 
				11 & 00, 11\\
				\hline
		\end{tabular}}
		\label{32211}
	\end{center}
\end{table}





From the strategies mentioned in subsection \ref{class}, one can easily check that the maximum classical success probability for this game is $0.75$ and one of the strategies to get this success probability is $a=0$ and $b=0$.

Similarly from the discussion of subsection \ref{quant}, one can easily check that the expression for quantum success probability of this game is of the form

\begin{align*}
&~ \frac{1}{4}\left[\frac{1}{2} + \frac{1}{2}\cos^2{\alpha} + \cos^2 {\beta} + \frac{1}{2}\sin^2 {\gamma} + \cos^2 {\delta}\right]\ \\
&= \frac{1}{2} + \frac{1}{4} [1+\cos{2\alpha}+2+2\cos{2\beta}+1-\cos{2\gamma}+2+2\cos{2\delta}]
\end{align*}
where $\alpha=(\theta_0 - \psi_0)$, $\beta=(\theta_0 - \psi_1)$, $\gamma=(\theta_1 - \psi_0)$ and 
$\delta=(\theta_1 - \psi_1)$. 

One can think of the cosines as the inner products between unit vectors. In that case, one can rewrite the above as

\begin{align*}
\frac{1}{4}[2 + \frac{1}{4}(u_{0}v_{0} + 2u_{0}v_{1} - u_{1}v_{0} + 2u_{1}v_{1})] 
\end{align*}
\begin{align*}
\le \frac{1}{4} \left[2 + \frac{1}{4}[ ||v_{0} + 2v_{1}|| +  ||-v_{0} + 2v_{1}||] \right]
\end{align*}

Let us assume, $\langle v_0, v_1 \rangle = a + ib$ and $\langle v_1, v_0 \rangle = a - ib$. Then $||v_0 + 2v_1|| = \sqrt{5 + 4a}$ and $||-v_0 + 2v_1|| = \sqrt{5 - 4a}$. It is easy to check that the expression $\sqrt{5 + 4a} + \sqrt{5 - 4a}$ attains maximum value for $a = 0$ and the corresponding maximum value is $2\sqrt{5}$ i.e., $ ( ||v_0 + 2v_1|| + ||-v_0 + 2v_1||) \leq 2\sqrt{5}$.

Hence the maximum winning probability 
$\le \frac{1}{4} \left[2 + \frac{1}{4} \times 2\sqrt{5} \right] \approx 0.78$.

This implies that one can find some values of $\theta_0, \theta_1, \psi_0$ and $\psi_1$ for which the game can be won with probability $0.78$ in quantum scenario. Similarly one can show the same upper bound for some other games of this group for which the maximum classical success probability is $0.75$.
Therefore the maximum separation for this class of games is $(0.78 - 0.75) \approx 0.03$.

From this discussion, it is clear that for all the games corresponding to partitions $4+2+1+1$ and $3+3+1+1$, there are no chances of getting quantum advantage. But for the partition $2+2+2+2$ and $3+2+2+1$, there are some games which provide quantum advantage with a separation around $0.103$ and $0.03$ respectively. A summary of these results are mentioned in Table \ref{partition8}.

\begin{table}[htb]
	\begin{center} 
			\begin{tabular}{ |c|c|c|c| } 
				\hline
				\multirow{2}{*}{Partitions} & Max. classical & Max. quantum & Corresponding\\
				~ & success prob. & success prob. & Separation\\
				\hline
				\hline
				\multirow{2}{*}{4+2+1+1} & 0.75 & 0.75 & NA\\
				~ & 1.0 & 1.0 & NA \\
				\hline
				\multirow{2}{*}{3+3+1+1} & 0.75 & 0.75 & NA\\
				~ & 1.0 & 1.0 & NA \\
				\hline
				\multirow{2}{*}{2+2+2+2} & 0.75 & 0.853 & 0.103\\
				~ & 1.0 & 1.0 & NA \\
				\hline 
				\multirow{2}{*}{3+2+2+1} & 0.75 & 0.78 & 0.03\\
				~ & 1.0 & 1.0 & NA \\
				\hline
			\end{tabular}
	\end{center}
	\caption{Analysis of partitions for $8$ successful outcomes}
	\label{partition8}
\end{table}

\subsection{Games Corresponding to $9$ Successful Outcomes}
Proceeding to the similar way as the analysis of the $8$ successful outcomes, the maximum classical and quantum success probabilities that one can achieve for each of the partitions of the $9$ successful outcomes are mentioned in the Table \ref{partition9}. From these results, one can easily check that quantum advantage can be achieved (with a separation around $0.042$) only for some of the games corresponding to the partition $3+3+2+1$. For simplicity, here we only consider a game (having quantum advantage) from the partition $3+3+2+1$ and analyze the performance.

\begin{table}[htb]
	\begin{center}
		\setlength{\tabcolsep}{0.35em}\scalebox{1}{
			\begin{tabular}{ |c|c|c|c| } 
				\hline
				\multirow{2}{*}{Partitions} & Max. classical & Max. quantum & Corresponding\\
				~ & success prob. & success prob. & Separation\\
				\hline
				\hline
				\multirow{2}{*}{4+3+1+1} & 0.75 & 0.75 & NA\\
				~ & 1.0 & 1.0 & NA \\
				\hline
				4+2+2+1 & 1.0 & 1.0 & NA\\
				\hline
				3+2+2+2 & 1.0 & 1.0 & NA\\
				\hline 
				\multirow{2}{*}{3+3+2+1} & 0.75 & 0.792 & 0.042\\
				~ & 1.0 & 1.0 & NA \\
				\hline
		\end{tabular}}
		\end{center}
		\caption{Analysis of partitions for $9$ successful outcomes} 
		\label{partition9}
\end{table}

\subsubsection{\it A Game for partition $3+3+2+1$ having quantum advantage}

From the results of table \ref{partition9}, it is clear that for the games having $9$ successful outcomes, quantum advantage can be achieved only for some of the games having partition $3+3+2+1$.
Here we consider the following game which can't be won with certainty in classical scenario.

\begin{table}[htb]
	\begin{center}
		\setlength{\tabcolsep}{0.35em}\scalebox{1}{
			\begin{tabular}{ |c|c| } 
				\hline
				Input & Corresponding output\\
				\hline
				\hline
				00 & 00, 11\\
				\hline
				01 & 00, 01, 10\\
				\hline
				10 & 11\\
				\hline 
				11 & 00, 01, 11\\
				\hline
		\end{tabular}}
	\end{center}
\end{table}


From the strategies mentioned in subsection \ref{class}, one can easily check that the maximum classical success probability for this game is $0.75$ and one of the strategies to get this success probability is $a=0$ and $b=0$.

Similarly from the discussion of subsection \ref{quant}, one can easily check that the expression for quantum success probability of this mentioned game is of the form

\begin{align*}
~ &~ \frac{1}{4}[\cos^2{\alpha} + \frac{1}{2} + \frac{1}{2} \sin^2 {\beta} + \frac{1}{2} + \frac{1}{2}\cos^2 {\gamma} + \frac{1}{2} \cos^2 {\delta}]\\
~ &= \frac{1}{4}[2 + \frac{1}{4} + \frac{1}{4}(2\cos{2\alpha} - \cos{2\beta} + \cos{2\gamma} + \cos{2\delta})]
\end{align*}
where $\alpha=(\theta_0 - \psi_0)$, $\beta=(\theta_0 - \psi_1)$, $\gamma=(\theta_1 - \psi_0)$ and 
$\delta=(\theta_1 - \psi_1)$. 

One can think of the cosines as the inner products between unit vectors. In that case, one can rewrite the above as

\begin{align*}
\frac{1}{4}[2 + \frac{1}{4} + \frac{1}{4}(2u_{0}v_{0} - u_{0}v_{1} + u_{1}v_{0} + u_{1}v_{1})]
\end{align*}
\begin{align*}
\le \frac{1}{4} \left[2 + \frac{1}{4} + \frac{1}{4}(||u_{0}|| ||2v_{0} - v_{1}|| + ||u_{1}|| ||v_{0} + v_{1}||) \right]
\end{align*}

Let us assume that $\langle v_0, v_1 \rangle = (a + ib)$ and $\langle v_1, v_0 \rangle = (a - ib)$. Then one can easily check that $||2v_0 - v_1|| = \sqrt{5 - 4a}$ and $||v_0 + v_1|| = \sqrt{2 + 2a}$. From these expressions, one can easily calculate that the expression $\sqrt{2 + 2a} + \sqrt{5 - 4a}$ attains maximum value for $a = -\frac{1}{4}$ and the corresponding maximum value is $\sqrt{6}+\sqrt{1.5}$. From this, the maximum winning probability in quantum scenario can be written as,
\begin{align*}
&~ \frac{1}{4} \left[2 + \frac{1}{4} + \frac{1}{4}(||u_{0}|| ||2v_{0} - v_{1}|| + ||u_{1}|| ||v_{0} + v_{1}||) \right]\\
&\leq \frac{1}{4} \left[\frac{9}{4} + \frac{1}{4} \times (\sqrt{6} + \sqrt{1.5})\right]\\
&\approx 0.792
\end{align*}

This implies that one can find some values of $\theta_0, \theta_1, \psi_0$ and $\psi_1$ for which the game can be won with a probability of $0.792$ in the quantum scenario. Similarly one can show the same upper bound for some other games of this group for which the maximum classical success probability is $0.75$.
Therefore the maximum separation for this class of games is $(0.792 - 0.75) \approx 0.042$.

From this discussion, it is clear that for all the games corresponding to partitions $4+3+1+1$, $4+2+2+1$ and $3+2+2+2$, there are no chances of getting quantum advantage. But for the partition $3+3+2+1$, there are some games which provide quantum advantage with a separation around $0.042$.


\subsection{Games Corresponding to $10$ Successful Outcomes}

Proceeding in similar way as the analysis of the $8$ successful outcomes, the maximum classical and quantum success probabilities that one can achieve for each of the partitions of the $10$ successful outcomes are mentioned in Table \ref{partition10}. From this result, one can easily check that the quantum advantage can be achieved only for some of the games corresponding to partition $3+3+3+1$ with a separation around $0.05$. For simplicity, here we consider one of these games having quantum advantage and analyze its performance in both classical and quantum scenarios.

\begin{table}[htb]
	\begin{center}
		\setlength{\tabcolsep}{0.35em}\scalebox{1}{
			\begin{tabular}{ |c|c|c|c| } 
				\hline
				\multirow{2}{*}{Partitions} & Max. classical & Max. quantum & Corresponding\\
				~ & success prob. & success prob. & Separation\\
				\hline
				\hline
				4+4+1+1 & 1.0 & 1.0 & NA\\
				\hline
				4+2+2+2 & 1.0 & 1.0 & NA\\
				\hline
				3+3+2+2 & 1.0 & 1.0 & NA\\
				\hline
				4+3+2+1 & 1.0 & 1.0 & NA\\
				\hline 
				\multirow{2}{*}{3+3+3+1} & 0.75 & 0.8 & 0.05\\
				~ & 1.0 & 1.0 & NA \\
				\hline
		\end{tabular}}
	\end{center}
	\caption{Analysis of partitions for $10$ successful outcomes} 
	\label{partition10}
\end{table}

\subsubsection{\it A game for partition $3+3+3+1$ having quantum advantage}

From the results of table \ref{partition10}, it is clear that for the games having $10$ successful outcomes, quantum advantage can be achieved only for some of the games having partition $3+3+3+1$.
Here we consider the following game which can't be won with certainty in classical scenario.

\begin{table}[htb]
	\begin{center}
		\setlength{\tabcolsep}{0.35em}\scalebox{1}{
			\begin{tabular}{ |c|c| } 
				\hline
				Input & Corresponding output\\
				\hline
				\hline
				00 & 00\\
				\hline
				01 & 00, 10, 11\\
				\hline
				10 & 00, 01, 11\\
				\hline 
				11 & 01, 10, 11\\
				\hline
		\end{tabular}}
	\end{center}
\end{table}


From the strategies mentioned in subsection \ref{class}, one can easily check that the maximum classical success probability for this game is $0.75$ and one of the strategies to get this success probability is $a=0$ and $b=0$.

Similarly from the discussion of subsection \ref{quant}, one can easily check that the expression for quantum success probability of this mentioned game is of the form

\begin{align*}
~ &~ \frac{1}{4}\left[\frac{1}{2}\cos^2{\alpha} + \frac{1}{2} + \frac{1}{2} \cos^2 {\beta} + \frac{1}{2} + \frac{1}{2}\cos^2 {\gamma} + \frac{1}{2} + \frac{1}{2} \sin^2 {\delta}\right]\\
~ &= \frac{1}{4}\left[\frac{5}{2} + \frac{1}{4}(\cos{2\alpha} + \cos{2\beta} + \cos{2\gamma} - \cos{2\delta})\right]
\end{align*}
where $\alpha=(\theta_0 - \psi_0)$, $\beta=(\theta_0 - \psi_1)$, $\gamma=(\theta_1 - \psi_0)$ and 
$\delta=(\theta_1 - \psi_1)$. 

One can easily think of the cosines as the inner products between unit vectors. In that case, one can rewrite the above expression as

\begin{align*}
\frac{1}{4}\left[\frac{5}{2} + \frac{1}{4}(u_{0}v_{0} + u_{0}v_{1} + u_{1}v_{0} - u_{1}v_{1})\right]
\end{align*}
\begin{align*}
\le \frac{1}{4} \left[\frac{5}{2} + \frac{1}{4}\left(||u_{0}|| ||v_{0} + v_{1}|| + ||u_{1}|| ||v_{0} - v_{1}||\right) \right]
\end{align*}
Now, $||u_{0}|| ||v_{0} + v_{1}|| + ||u_{1}|| ||v_{0} - v_{1}|| \le ||v_{0} + v_{1}|| + ||v_{0} - v_{1}|| \le 2\sqrt{2}$.

Hence the winning probability in quantum scenario
$\le \frac{1}{4} \left[\frac{5}{2} + \frac{1}{4} \times 2\sqrt{2} \right] \approx 0.80$.

This implies that one can find some values of $\theta_0, \theta_1, \psi_0$ and $\psi_1$ for which the game can be won with a probability of $0.80$ in the quantum scenario. Similarly one can show the same upper bound for some other games of this group for which the maximum classical success probability is $0.75$.
Therefore the maximum separation for this class of games is $(0.80 - 0.75) \approx 0.05$.

From this discussion, it is clear that for all the games corresponding to partitions $4+4+1+1$, $4+2+2+2$, $3+2+2+2$ and $4+3+2+1$, there are no chances of getting quantum advantage. But for the partition $3+3+3+1$, there are some games which provide quantum advantage with a separation around $0.05$.


\subsection{Games Corresponding to $11$ Successful Outcomes}

Proceeding in similar way as the analysis of the $8$ successful outcomes, the maximum classical and quantum success probabilities that one can achieve for each of the partitions of the $11$ successful outcomes are mentioned in the Table \ref{partition11}. From this result, one can easily check that there are no games corresponding to $11$ successful outcomes for which quantum advantage can be achieved.

\begin{table}[htb]
	\begin{center}
		\setlength{\tabcolsep}{0.35em}\scalebox{1}{
			\begin{tabular}{ |c|c|c|c| } 
				\hline
				\multirow{2}{*}{Partitions} & Max. classical & Max. quantum & Corresponding\\
				~ & success prob. & success prob. & Separation\\
				\hline
				\hline
				4+4+2+1 & 1.0 & 1.0 & NA\\
				\hline
				4+3+3+1 & 1.0 & 1.0 & NA\\
				\hline
				4+3+2+2 & 1.0 & 1.0 & NA\\
				\hline 
				3+3+3+2 & 1.0 & 1.0 & NA\\
				\hline
		\end{tabular}}
	\end{center}
	\caption{Analysis of partitions for $11$ successful outcomes} 
	\label{partition11}
\end{table}

So for all the games corresponding to $11$ successful outcomes, there are no chances of getting any advantage in quantum success probability as compared to the classical one.

\subsection{Games Corresponding to $12$ or More Successful Outcomes}

One can easily verify that each of the partitions for $12$ successful outcomes is an extension of some partitions corresponding to $11$ successful outcomes. As all the games corresponding to $11$ successful outcomes can be won classically with certainty, there is no chance of getting quantum advantage for any of the games having $12$ successful outcomes. Similarly one can also argue the same statement for $13$ or more successful outcomes.

For this reason, the games having $12$ or more successful outcomes can't achieve quantum advantage.

\subsection{Games Corresponding to $7$ Successful Outcomes}

Proceeding in similar way as the analysis of the $8$ successful outcomes, the maximum classical and quantum success probabilities that one can achieve for each of the partitions of the $7$ successful outcomes are mentioned in the Table \ref{partition7}. From these results, one can easily check that quantum advantage can be achieved only for some of the games corresponding to partition $2+2+2+1$ with a separation around $0.012$. For simplicity, here we consider one of these games having quantum advantage and analyze its performance.\\

\begin{table}[htb]
	\begin{center}
		\setlength{\tabcolsep}{0.35em}\scalebox{1}{
			\begin{tabular}{ |c|c|c|c| } 
				\hline
				\multirow{2}{*}{Partitions} & Max. classical & Max. quantum & Corresponding\\
				~ & success prob. & success prob. & Separation\\
				\hline
				\hline
				\multirow{2}{*}{3+2+1+1} & 0.75 & 0.75 & NA\\
				~ & 1.0 & 1.0 & NA \\
				\hline
				\multirow{2}{*}{2+2+2+1} & 0.75 & 0.762 & 0.012\\
				~ & 1.0 & 1.0 & NA \\
				\hline
		\end{tabular}}
	\end{center}
	\caption{Analysis of partitions for $7$ successful outcomes} 
	\label{partition7}
\end{table}

\subsubsection{\it A game for partition $2+2+2+1$ having quantum advantage}



From the results of table \ref{partition7}, it is clear that for the games having $7$ successful outcomes, quantum advantage can be achieved only for some of the games having partition $2+2+2+1$.
Here we consider the following game which can't be won with certainty in the classical scenario.

\begin{table}[htb]
	\begin{center}
		\setlength{\tabcolsep}{0.35em}\scalebox{1}{
			\begin{tabular}{ |c|c| } 
				\hline
				Input & Corresponding output\\
				\hline
				\hline
				00 & 00, 11\\
				\hline
				01 & 01, 10\\
				\hline
				10 & 11\\
				\hline 
				11 & 00, 11\\
				\hline
		\end{tabular}}
	\end{center}
\end{table}

From the strategies mentioned in subsection \ref{class}, one can easily check that the maximum classical success probability for this game is $0.75$ and one of the strategies to get this success probability is $a=1$ and $b=1$.

Similarly from the discussion of subsection \ref{quant}, one can easily check that the expression for quantum success probability of this mentioned game is of the form



\begin{align*}
~ &~ \frac{1}{4}\left[\cos^2{\alpha} + \sin^2 {\beta} + \frac{1}{2}\cos^2 {\gamma} + \cos^2 {\delta}\right]\\
~ &= \frac{1}{4}\left[\frac{7}{4} + \frac{1}{4}(2\cos{2\alpha} - 2\cos{2\beta} + \cos{2\gamma} + 2\cos{2\delta})\right]
\end{align*}
where $\alpha=(\theta_0 - \psi_0)$, $\beta=(\theta_0 - \psi_1)$, $\gamma=(\theta_1 - \psi_0)$ and 
$\delta=(\theta_1 - \psi_1)$. 

One can easily think of the cosines as the inner products between unit vectors. In that case, one can rewrite the above expression as
\begin{align*}
\frac{1}{4}\left[\frac{7}{4} + \frac{1}{4}(2u_{0}v_{0} - 2u_{0}v_{1} + u_{1}v_{0} + 2u_{1}v_{1})\right]
\end{align*}
\begin{align*}
\le \frac{1}{4} \left[\frac{7}{4} + \frac{1}{4}\left(||u_{0}|| ||2v_{0} - 2v_{1}|| + ||u_{1}|| ||v_{0} + 2v_{1}||\right) \right]
\end{align*}

Let us assume that $\langle v_0, v_1 \rangle = (a + ib)$ and $\langle v_1, v_0 \rangle = (a - ib)$. Then one can easily check that $||2v_0 - 2v_1|| = 2||v_0 - v_1|| = 2\sqrt{2 - 2a}$ and $||v_0 + 2v_1|| = \sqrt{5 + 4a}$. From these expressions, one can easily calculate that the expression $2\sqrt{2 - 2a} + \sqrt{5 + 4a}$ attains maximum value for $a = -\frac{1}{2}$ and the corresponding maximum value is $3\sqrt{3}$. From this, the maximum winning probability in quantum scenario can be written as,

\begin{align*}
&~ \frac{1}{4} \left[\frac{7}{4} + \frac{1}{4}\left(||u_{0}|| ||2v_{0} - 2v_{1}|| + ||u_{1}|| ||v_{0} + 2v_{1}||\right) \right]\\
&\leq \frac{1}{4} \left[\frac{7}{4} + \frac{1}{4} \times 3\sqrt{3}\right]\\
&\approx 0.762
\end{align*}

This implies that one can find some values of $\theta_0, \theta_1, \psi_0$ and $\psi_1$ for which the game can be won with a probability of $0.762$ in the quantum scenario. Similarly one can explore that the same upper bound can be achieved for all the other games of this group for which quantum advantage can be achieved and the maximum classical success probability is $0.75$.
Therefore the maximum separation for this class of games is $(0.762 - 0.75) \approx 0.012$.


From this discussion, it is clear that for all the games corresponding to partition $3+2+1+1$, there are no chances of getting quantum advantage. But for the partition $2+2+2+1$, there are some games which provide quantum advantage with a separation around $0.012$.

\subsection{Games Corresponding to $6$ Successful Outcomes}

Proceeding in similar way as the analysis of the $8$ successful outcomes, the maximum classical and quantum success probabilities that one can achieve for each of the partitions of the $6$ successful outcomes are mentioned in the Table \ref{partition6}. From these results, one can easily check that quantum advantage can be achieved only for some of the games corresponding to partition $3+1+1+1$ with a separation around $0.05$. Here we consider one of these games having quantum advantage and analyze its performance in both classical and quantum scenarios.

\begin{table}[htb]
	\begin{center}
		\setlength{\tabcolsep}{0.35em}\scalebox{1}{
			\begin{tabular}{ |c|c|c|c| } 
				\hline
				\multirow{2}{*}{Partitions} & Max. classical & Max. quantum & Corresponding\\
				~ & success prob. & success prob. & Separation\\
				\hline
				\hline
				\multirow{3}{*}{3+1+1+1} & 0.5 & 0.55 & 0.05\\
				~ & 0.75 & 0.75 & NA\\
				~ & 1.0 & 1.0 & NA \\
				\hline
				\multirow{2}{*}{2+2+1+1} & 0.75 & 0.75 & NA\\
				~ & 1.0 & 1.0 & NA \\
				\hline
		\end{tabular}}
	\end{center}
	\caption{Analysis of partitions for $6$ successful outcomes} 
	\label{partition6}
\end{table}

\subsubsection{\it A game for partition $3+1+1+1$ having quantum advantage}


From the results of table \ref{partition6}, it is clear that for the games having $6$ successful outcomes, quantum advantage can be achieved only for some of the games having partition $3+1+1+1$.
Here we consider the following game which can't be won with certainty in the classical scenario.

\begin{table}[htb]
	\begin{center}
		\setlength{\tabcolsep}{0.35em}\scalebox{1}{
			\begin{tabular}{ |c|c| } 
				\hline
				Input & Corresponding output\\
				\hline
				\hline
				00 & 00, 01, 10\\
				\hline
				01 & 11\\
				\hline
				10 & 01\\
				\hline 
				11 & 10\\
				\hline
		\end{tabular}}
	\end{center}
\end{table}

From the strategies mentioned in subsection \ref{class}, one can easily check that the maximum classical success probability for this game is $0.5$ and one of the strategies to get this success probability is $a=0$ and $b=1$.

Similarly from the discussion of subsection \ref{quant}, one can easily check that the expression for quantum success probability of this mentioned game is of the form
\begin{align*}
~ &~ \frac{1}{4}\left[\frac{1}{2} + \frac{1}{2}\sin^2{\alpha} + \frac{1}{2} \cos^2 {\beta} + \frac{1}{2}\sin^2 {\gamma} + \frac{1}{2} \sin^2 {\delta}\right]\\
~ &= \frac{1}{4}\left[\frac{3}{2} + \frac{1}{4}(-\cos{2\alpha} + \cos{2\beta} - \cos{2\gamma} - \cos{2\delta})\right]
\end{align*}
where $\alpha=(\theta_0 - \psi_0)$, $\beta=(\theta_0 - \psi_1)$, $\gamma=(\theta_1 - \psi_0)$ and 
$\delta=(\theta_1 - \psi_1)$. 

One can easily think of the cosines as the inner products between two unit vectors. In that case, one can rewrite the above expression as
\begin{align*}
\frac{1}{4}\left[\frac{3}{2} + \frac{1}{4}(-u_{0}v_{0} + u_{0}v_{1} - u_{1}v_{0} - u_{1}v_{1})\right]
\end{align*}
\begin{align*}
\le \frac{1}{4} \left[\frac{3}{2} + \frac{1}{4}\left(||u_{0}|| ||v_{0} - v_{1}|| + ||u_{1}|| ||v_{0} + v_{1}||\right) \right]
\end{align*}
Now, $||u_{0}|| ||v_{0} - v_{1}|| + ||u_{1}|| ||v_{0} + v_{1}|| \le ||v_{0} - v_{1}|| + ||v_{0} + v_{1}|| \le 2\sqrt{2}$.

Hence the winning probability 
$\le \frac{1}{4} \left[\frac{3}{2} + \frac{1}{4} \times 2\sqrt{2} \right] \approx 0.55$.

This implies that one can find some values of $\theta_0, \theta_1, \psi_0$ and $\psi_1$ for which the game can be won with a probability of $0.55$ in the quantum scenario. Similarly one can explore that the same upper bound can be achieved for all the other games of this group for which quantum advantage can be achieved and the maximum classical success probability is $0.5$.
Therefore the maximum separation for this class of games is $(0.55 - 0.5) \approx 0.05$.

From this discussion, it is clear that for all the games corresponding to partition $2+2+1+1$, there are no chances of getting quantum advantage. But for the partition $3+1+1+1$, there are some games which provide quantum advantage with a separation around $0.05$. 

\subsection{Games Corresponding to $5$ Successful Outcomes}

Proceeding in similar way as the analysis of the $8$ successful outcomes, the maximum classical and quantum success probabilities that one can achieve for each of the partitions of the $5$ successful outcomes are mentioned in the Table \ref{partition5}. From these results, one can easily check that quantum advantage can be achieved for the games corresponding to partition $2+1+1+1$ with a separation around $0.042$. Here we consider one of these games having quantum advantage and analyze its performance in both classical and quantum scenarios.\\

\begin{table}[htb]
	\begin{center}
		\setlength{\tabcolsep}{0.35em}\scalebox{1}{
			\begin{tabular}{ |c|c|c|c| } 
				\hline
				\multirow{2}{*}{Partitions} & Max. classical & Max. quantum & Corresponding\\
				~ & success prob. & success prob. & Separation\\
				\hline
				\hline
				\multirow{3}{*}{2+1+1+1} & 0.5 & 0.542 & 0.042\\
				~ & 0.75 & 0.75 & NA\\
				~ & 1.0 & 1.0 & NA \\
				\hline
		\end{tabular}}
	\end{center}
	\caption{Analysis of partitions for $5$ successful outcomes} 
	\label{partition5}
\end{table}

\subsubsection{\it A game for partition $2+1+1+1$ having quantum advantage}

From the results of table \ref{partition5}, it is clear that for the games having $5$ successful outcomes, quantum advantage can be achieved only for some of the games having partition $2+1+1+1$.
Here we consider the following game which can't be won with certainty in the classical scenario.


\begin{table}[htb]
	\begin{center}
		\setlength{\tabcolsep}{0.35em}\scalebox{1}{
			\begin{tabular}{ |c|c| } 
				\hline
				Input & Corresponding output\\
				\hline
				\hline
				00 & 01, 10\\
				\hline
				01 & 11\\
				\hline
				10 & 01\\
				\hline 
				11 & 10\\
				\hline
		\end{tabular}}
	\end{center}
\end{table}

From the strategies mentioned in subsection \ref{class}, one can easily check that the maximum classical success probability for this game is $0.5$ and one of the strategies to get this success probability is $a=0$ and $b=1$.

Similarly from the discussion of subsection \ref{quant}, one can easily check that the expression for quantum success probability of this mentioned game is of the form

\begin{align*}
~ &~ \frac{1}{4}\left[\sin^2{\alpha} + \frac{1}{2} \cos^2 {\beta} + \frac{1}{2}\sin^2 {\gamma} + \frac{1}{2} \sin^2 {\delta}\right]\\
~ &= \frac{1}{4}\left[\frac{5}{4} + \frac{1}{4}(-2\cos{2\alpha} + \cos{2\beta} - \cos{2\gamma} - \cos{2\delta})\right]
\end{align*}
where $\alpha=(\theta_0 - \psi_0)$, $\beta=(\theta_0 - \psi_1)$, $\gamma=(\theta_1 - \psi_0)$ and 
$\delta=(\theta_1 - \psi_1)$. \\ 
As one can think of the cosines as the inner products between unit vectors, the above expression can be rewritten as,
\begin{align*}
\frac{1}{4}\left[\frac{5}{4} + \frac{1}{4}(-2u_{0}v_{0} + u_{0}v_{1} - u_{1}v_{0} - u_{1}v_{1})\right]
\end{align*}
\begin{align*}
\le \frac{1}{4} \left[\frac{5}{4} + \frac{1}{4}\left(||-2v_{0} + v_{1}|| +  ||v_{0} + v_{1}||\right)\right]
\end{align*}

Let us assume that $\langle v_0, v_1 \rangle = (a + ib)$ and $\langle v_1, v_0 \rangle = (a - ib)$. Then one can easily check that $||-2v_0 + v_1|| = \sqrt{5 - 4a}$ and $||v_0 + v_1|| = \sqrt{2 + 2a}$. From these expressions, one can easily calculate that the expression $\sqrt{2 + 2a} + \sqrt{5 - 4a}$ attains maximum value for $a = -\frac{1}{4}$ and the corresponding maximum value is $\sqrt{6}+\sqrt{1.5}$. From this, the maximum winning probability in quantum scenario can be written as,
\begin{align*}
&~ \frac{1}{4} \left[\frac{5}{4} + \frac{1}{4}\left(||-2v_{0} + v_{1}|| + ||v_{0} + v_{1}||\right)\right]\\
&\leq \frac{1}{4} \left[\frac{5}{4} + \frac{1}{4} \times (\sqrt{6} + \sqrt{1.5})\right]\\
&\approx 0.542
\end{align*}

This implies that one can find some values of $\theta_0, \theta_1, \psi_0$ and $\psi_1$ for which the game can be won with a probability of $0.542$ in the quantum scenario. Similarly one can explore that the same upper bound can be achieved for all the other games of this group for which quantum advantage can be achieved and the maximum classical success probability is $0.5$.
Therefore the maximum separation for this class of games is $(0.542 - 0.5) \approx 0.042$.


From this discussion, it is clear that for all the games corresponding to partition $2+1+1+1$, there are some games which provide quantum advantage with a separation around $0.042$. 

\subsection{Games Corresponding to $4$ Successful Outcomes}

Proceeding in similar way as the analysis of the $8$ successful outcomes, the maximum classical and quantum success probabilities that one can achieve for the partition of the $4$ successful outcomes are mentioned in the Table \ref{partition4}. From these results, one can easily check that there are no games corresponding to $4$ successful outcomes for which quantum advantage can be achieved.

\begin{table}[htb]
	\begin{center}
		\setlength{\tabcolsep}{0.35em}\scalebox{1}{
			\begin{tabular}{ |c|c|c|c| } 
				\hline
				\multirow{2}{*}{Partitions} & Max. classical & Max. quantum & Corresponding\\
				~ & success prob. & success prob. & Separation\\
				\hline
				\hline
				\multirow{3}{*}{1+1+1+1} & 0.5 & 0.5 & NA\\
				~ & 0.75 & 0.75 & NA \\
				~ & 1.0 & 1.0 & NA \\
				\hline
		\end{tabular}}
	\end{center}
	\caption{Analysis of partition for $4$ successful outcomes} 
	\label{partition4}
\end{table}

So for all the games corresponding to $4$ successful outcomes, there are no chances of getting any advantage in quantum success probability as compared to the classical one.

\subsection{Games Corresponding to $3$ or Less Successful Outcomes}

One cannot divide 3 or less number of successful  outcomes into four parts such that each part has atleast one outcome. Hence for these class of games, one can easily argue from the discussion in subsection \ref{quant} that the maximum quantum success probability is always less than $0.5$ and there is no chance of getting quantum advantage.

\begin{center}
\begin{table}[htbp]
		\setlength{\tabcolsep}{0.005em}\scalebox{1}{
			\begin{tabular}{ |c|c|c|c|c|c|c|c| } 
				\hline
				No. of & Partition & A game corr. & Max. & Corr. & Max. & Corr. & Max.\\
				succ. & with quant. & to the partition & classical & classical & quantum & quantum & Sepa-\\
				outcome & advantage & (in ANF form) & succ. & strategy & succ. &  strategy & ration\\
				~ & ~ & ~ &  prob. & ~ & prob. & ~ & ~\\
				\hline
				10 & 3+3+3+1 & $a \oplus b \oplus xy \oplus xb \oplus ya$ & 0.75 & a=0,& 0.80 & $\theta_0 = 0, \theta_1 = \frac{\pi}{4}$ & 0.05\\
				~ & ~ & $\oplus ab \oplus xya \oplus xyb$ & ~ & b=0 & ~ & $\psi_0 = \frac{\pi}{8}, \psi_1 = \frac{7\pi}{8}$ & ~ \\
				\hline
				~ & ~ & $x \oplus a \oplus b \oplus xy \oplus xa$ & ~ & a=0, & ~ & $\theta_0 = 0,\theta_1 = \frac{21\pi}{100}$ & ~\\
				9 & 3+3+2+1 & $\oplus xb \oplus ya \oplus yb \oplus xyb$ & 0.75 & b=0 & 0.792 & $\psi_0 = \frac{\pi}{8},\psi_1 = \frac{7\pi}{20}$ & 0.042\\
				~ & ~ & $\oplus xab \oplus yab \oplus xyab$ & ~ & ~ & ~ & ~ & ~\\
				\hline
				8 & 2+2+2+2 & $a \oplus b \oplus xy$ & 0.75 & a=0, & 0.853 & $\theta_0 = 0,\theta_1 = \frac{\pi}{4}$ & 0.103\\
				~ & ~ & ~ & ~ & b=0 & ~ & $\psi_0 = \frac{\pi}{8},\psi_1 = \frac{7\pi}{8}$ & ~\\
				\hline 
				~ & ~ & $x \oplus a \oplus xy \oplus xa$ & ~ &  a=0, & ~ & $\theta_0 = 0,\theta_1 = \frac{7\pi}{50}$ & ~\\
				8 & 3+2+2+1 & $\oplus xb \oplus yb \oplus ab$ & 0.75 & b=0 & 0.78 & $\psi_0 = \frac{5\pi}{6},\psi_1 = \frac{7\pi}{100}$ & 0.03\\
				~ & ~ & $\oplus xya \oplus xyb \oplus yab$ & ~ & ~ & ~ & ~ & ~\\
				\hline
				~ & ~ & $x \oplus y \oplus a \oplus b$ & ~ & a=1, & ~ & $\theta_0 = 0,\theta_1 = \frac{\pi}{3}$ & ~\\
				7 & 2+2+2+1 & $\oplus xa \oplus xb \oplus xya$ & 0.75 & b=1 & 0.762 & $\psi_0 = \frac{2\pi}{25},\psi_1 = \frac{21\pi}{50}$ & 0.012\\
				~ & ~ & $\oplus xyb \oplus xab \oplus xyab$ & ~ & ~ & ~ & ~ & ~\\
				\hline
				6 & 3+1+1+1 & $x \oplus y \oplus xy \oplus xb$ & 0.5 & a=0, & 0.55 & $\theta_0 = 0,\theta_1 = \frac{\pi}{4}$ & 0.05\\
				~ & ~ & $\oplus ab \oplus xya \oplus xyb$ & ~ & b=1 & ~ & $\psi_0 = \frac{5\pi}{8},\psi_1 = \frac{7\pi}{8}$ & ~\\
				\hline
				~ & ~ & $1 \oplus a \oplus b \oplus xa$ & ~ & a=0, & ~ & $\theta_0 = 0,\theta_1 = \frac{21\pi}{100}$ & ~\\
				5 & 2+1+1+1 & $\oplus ya \oplus yb \oplus xab$ & 0.5 & b=1 & 0.542 & $\psi_0 = \frac{14\pi}{25},\psi_1 = \frac{17\pi}{20}$ & 0.042\\
				~ & ~ & $\oplus yab \oplus xyab$ & ~ & ~ & ~ & ~ & ~\\
				\hline
		\end{tabular}}
	\caption{List of partitions and the corresponding nonlocal games (in ANF form) which offer quantum advantage} 
	\label{separation}
\end{table}
\end{center}


\section{Conclusion}
\label{concl}
 In our analysis, we found only seven partitions (over all possible games having at least one successful outcome for each possible input) such that the games corresponding to those partitions offer a quantum advantage. The maximum classical and the maximum quantum success probabilities for the games corresponding to each of those partitions are mentioned in Table \ref{separation}. We also mention an example of such a game (in Algebraic Normal Form) for each of those partitions. It is well known that the CHSH game is used to certify untrusted devices in the device-independent scenario. It is also known that the required sample size for device-independent testing is inversely proportional to the success probability of the corresponding nonlocal game. Although the maximum success probability for the CHSH game using quantum resources is less than $1$ (around $0.85$), so far no other two-party nonlocal game is used for device-independent testing. To the best of our knowledge, it was also unknown whether there exists any other binary input binary output two-party nonlocal game which offers a quantum advantage. To answer all these questions, in this article, we explore the performance of all possible binary input binary output two-party nonlocal games in terms of partitions of the total number of successful outcomes to check whether there exist any such games which offer a quantum advantage with maximum quantum success probability greater than $0.85$. From our analysis, we found that there are some binary input binary output two-party nonlocal games (other than the CHSH game) that offer quantum advantage but the CHSH game has the maximum quantum success probability (also with a maximum separation of around $0.1$) among all these games. Further study for three (or more) party nonlocal games will be an interesting research work in this direction.

\end{document}